\documentclass[11pt,a4paper]{amsart}
\usepackage{a4wide}
\usepackage{epsfig}
\usepackage{amsmath}
\usepackage{graphicx}
\usepackage{verbatim}
\usepackage{amsmath}
\usepackage{latexsym}
\usepackage{amssymb}
\usepackage{amscd}
\usepackage{hyperref}
\usepackage{amsthm}

\newcommand{\CC}{\mathbb C}

\newcommand{\ZZ}{\mathbb Z}

\newcommand{\PP}{\mathbb P}

\newcommand{\cL}{\mathcal L}
\newcommand{\cM}{\mathcal M}

\newcommand{\hh}{\hat{h}}
\newcommand{\hht}{\hat{t}}

\newcommand{\p}{\partial}
\newcommand{\ve}{\varepsilon}

\newtheorem{theorem}{Theorem}
\newtheorem{prop}[theorem]{Proposition}
\newtheorem{lemma}[theorem]{Lemma}

\newtheorem{example}[theorem]{Example}

\newtheorem{Def}[theorem]{Definition}

\begin{document}

\title[Modular Frobenius manifolds and their invariant flows]
{Modular Frobenius manifolds and their invariant flows}

\author{Ewan K. Morrison, Ian A.B. Strachan}
\date{3$^{\rm rd}$ June 2010}
\address{Department of Mathematics\\ University of Glasgow\\
Glasgow G12 8QQ\\ U.K.}

\email{e.morrison@maths.gla.ac.uk, i.strachan@maths.gla.ac.uk}


\begin{abstract}
The space of Frobenius manifolds has a natural involutive symmetry on it: there exists a map $I$ which send a Frobenius manifold
to another Frobenius manifold. Also, from a Frobenius manifold one may construct a so-called almost dual Frobenius manifold which
satisfies almost all of the axioms of a Frobenius manifold. The action of $I$ on the almost dual manifolds is studied, and the action of
$I$ on objects such as periods, twisted periods and flows is studied. A distinguished class of Frobenius manifolds sit at the
fixed point of this involutive symmetry, and this is made manifest in certain modular properties of the various structures.
In particular, up to a simple
reciprocal transformation, for this class of modular Frobenius manifolds, the flows are invariant under the action of $I\,.$
\end{abstract}

\maketitle

\tableofcontents

\section{Introduction}

The main property of the Chazy equation

\begin{equation}
\gamma'''(\tau)=6\gamma(\tau) \gamma''(\tau)-9{\gamma'(\tau)}^2 \,,
\label{chazy}
\end{equation}
is its modularity property, that is, its invariance under the action of the group
$SL(2,\mathbb{Z})\,:$
\begin{eqnarray*}
\tau & \mapsto & \frac{a\tau + b}{c\tau + d}, \quad a\,,b\,,c\,,d \in \mathbb{Z}\,, ad-bc =1,  \\
\gamma(\tau) & \mapsto & (c\tau + d)^2\gamma(\tau)+2c(c\tau + d).
\end{eqnarray*}
The aim of this paper is to study integrable equations of hydrodynamic type
\[
\frac{\partial t^i}{\partial T} = M^i_j({\bf t}) \frac{\partial t^j}{\partial X}\,, \qquad\qquad i\,,j=1\,,\ldots\,, N
\]
which have a similar invariance under certain multivariable extensions of such modular transformations.

The origins of such systems comes from the study of certain Frobenius manifolds \cite{dubrovin1}.
Given an arbitrary Frobenius manifold one can associate to it natural families of hydrodynamic systems, and the systems that will
be studied arise from this construction. Before discussing the symmetry properties of these integrable hydrodynamic
equations we first examine the modularity properties of the underlying Frobenius manifold.
The space of Frobenius manifolds possesses various symmetries and the
main object of study in this paper is an involutive symmetry $I\,,$ also known as a type II-symmetry, which maps solutions $F$ of the underlying WDVV equations
to a new solution $\widehat F$:
\[
I\,:\, F \longrightarrow {\widehat F}\,.
\]
Up to some trivial redefinitions of variables by changes of sign, $I^2=id\,,$ the identity transformation
\footnote{For convenience we refer to a Frobenius manifold $F$ rather than to the Frobenius manifold whose prepotential is $F\,.$}.
However, certain special Frobenius manifolds lie at the fixed
point of this involution and these inherit special modularity properties: for convenience such manifolds will be referred
to as modular Frobenius manifolds, and a characterization of these, in terms of spectral data, will be given. Examples of such
modular Frobenius manifolds include:

\begin{itemize}

\item[$\bullet$] Jacobi group orbit spaces \cite{B,iabs2};

\item[$\bullet$] the unfolding spaces of the three simple elliptic singularities \cite{iabs3} ${\widetilde E}_{6,7,8}\,.$

\end{itemize}

In the first part of the paper we combine the study of $I$ with the notion of an almost dual Frobenius manifold. Given a Frobenius
manifold, and hence a solution $F$ of the underlying WDVV equations, Dubrovin \cite{dubrovin2} showed how one can construct a new solution $F^\star$
which satisfies most (but crucially not all) of the axioms of a Frobenius manifold. We denote this transformation by $\star:$
\[
F \stackrel{\star}{\longrightarrow} F^\star
\]
Combining this with the symmetry $I$ leads to the following picture
\[
\begin{array}{ccc}
F&\overset{I}{\longrightarrow}&
{\hat F}\\
\downarrow & & \downarrow \\
F^\star
&&
{\hat F}^\star
\end{array}
\]
In section 4 we study the induced involutive symmetry
\[
F^\star \overset{I^\star}{\longrightarrow} {\hat F}^\star\,.
\]
in terms of almost duality. It turns out that the action $I^\star$ is different for modular and non-modular Frobenius manifolds
(the action of type I-symmetries were discussed in \cite{RS2}).

The study of the action of the modular group of the associated hydrodynamic flows will be done in two stages:
the action of the symmetry $I$ on the tensors $M^i_j({\bf t})$ will first be given for an arbitrary Frobenius manifold
(similar results have also been recently obtained by Dingdian and Zhang \cite{zhang}) and then specialized to those tensors arising from modular
Frobenius manifolds. This gives a new symmetry of the corresponding hydrodynamic flows.

Flow equations in the dual picture can also be written down (this being essentially a Muira-type transformation) and the action
of $I^\star$ on these flows can then be studied. One of the advantages of this dual picture is that is it easier to write down
general examples of such flows. Thus large classes of explicit modular invariant flows may be written down explicitly.

\medskip

\section{Preliminaries}\label{pre}
\subsection{Frobenius Manifolds and the WDVV Equations}
Frobenius manifolds were introduced as a way to give a geometric understanding to solutions of the Witten-Dijkraaf-Verlinde-Verlinde (WDVV) equations,
\begin{equation}\label{wdvv}
\frac{\partial^3 F}{\partial t^{\alpha} \partial t^{\beta} \partial t^{\lambda}}\eta^{\lambda \mu} \frac{\partial^3 F}{\partial t^{\mu} \partial t^{\delta} \partial t^{\gamma}} = \frac{\partial^3 F}{\partial t^{\delta} \partial t^{\beta} \partial t^{\lambda}}\eta^{\lambda \mu} \frac{\partial^3 F}{\partial t^{\mu} \partial t^{\alpha} \partial t^{\gamma}}
\end{equation}
for some quasihomogeneous function $F(t)$. Throughout this paper $\eta^{\alpha\beta}$ will be defined via $\eta^{\alpha\beta}\eta_{\beta\kappa}=\delta^{\alpha}_{\kappa}$ where
\begin{equation}
\eta_{\alpha\beta} = \frac{\partial^3 F}{\partial t^1 \partial t^{\alpha} \partial t^{\beta}}
\end{equation}
is constant and non-degenerate.
We recall briefly how to establish the correspondence between Frobenius manifolds and solutions of WDVV.
\begin{Def} The triple $(A,\circ, \eta)$ is a \emph{Frobenius algebra} if:
\begin{itemize}
\item[1.]{$(A, \circ)$ is a commutative associative algebra over $\CC$ with unity $e$;}
\item[2.]{The bilinear pairing $\eta$ and mutiplication $\circ$ satisfy the following \emph{Frobenius condition}
\begin{equation*}
\eta( X\circ Y, Z ) = \eta( X, Y \circ Z )\,, \quad X, Y, Z \in A.
\end{equation*}}
\end{itemize}
\end{Def}
\noindent With this one may define a Frobenius manifold.
\begin{Def} Let $\cM$ be a smooth manifold. $\cM$ is called a \emph{Frobenius manifold} if each tangent space $T_t\cM$ is equipped with the structure of a Frobenius algebra varying smoothly with $t\in \cM$, and further
\begin{itemize}
\item[1.]{The invariant inner product $\eta$ defines a flat metric on $\cM$. }
\item[2.]{The unity vector field is covariantly constant with respect to the Levi-Civit\'{a} connection for $\eta$,
\begin{equation}\label{fm1}
{}^\eta\nabla e =0.
\end{equation}}
\item[3.]{Let
\begin{equation}
c(X,Y,Z):=\eta( X\circ Y, Z)\,, \quad X, Y, Z \in T_t\cM,
\end{equation}
then the $(0,4)$ tensor ${}^\eta\nabla_Wc(X,Y,Z)$ is totally symmetric. }
\item[4.]{There exists a vector field $E\in\Gamma(T\cM)$ such that $\nabla\nabla E=0$ and
\begin{equation}
\cL_E \eta = (2-d)\eta, \quad \cL_E \circ = \circ, \quad \cL_E e = -e.
\end{equation}}
$E$ is called the Euler vector field.
\end{itemize}
\end{Def}
Condition 1 implies there exist a choice of coordinates $(t^1, ..., t^N)$ such that the Gram matrix $\eta_{\alpha\beta}=\langle \partial_{\alpha}, \partial_{\beta} \rangle$ is constant. Furthermore, this may be done in such a way that $e=\partial_1$. In such a coordinate system, partial and covariant derivatives coincide, and condition 3 becomes $c_{\alpha\beta\gamma,\kappa} =c_{\alpha\beta\kappa,\gamma}$. Successive applications of the Poincar\'{e} lemma then implies local existence of a function $F(t)$ called the \emph{free energy} of the Frobenius manifold such that
\begin{equation}
c_{\alpha\beta\gamma} = \frac{\partial^3 F(t)}{\partial t^{\alpha}\partial t^{\beta} \partial t^{\gamma}}.
\end{equation}
Since $\eta(X,Y) = \eta (e\circ X, Y) = c(e, X, Y)$, we have
\begin{equation}
\eta_{\alpha\beta} = c_{1\alpha\beta}.
\end{equation}
Defining $(\eta_{\alpha\beta})^{-1}=\eta^{\alpha\beta}$, the components of $\circ$ are given by $c^{\alpha}_{\beta\gamma}=\eta^{\alpha\varepsilon}c_{\varepsilon\beta\gamma}$. Associativity of $\circ$ is then equivalent to \eqref{wdvv}. Condition 4 leads to requiring $F$ be quasihomogeneous
\begin{equation}
\cL_EF = E(F) = d_F \cdot F = (3-d)F \quad \mbox{modulo quadratic terms}.
\end{equation}
We remark further that in the case where the grading operator $\nabla E$ of the Frobenius algebras is diagonalizable, the Euler vector field may be reduced to the form
\begin{equation}\label{euvect1}
E(t) = \sum_{\alpha}(1-q_\alpha)t^{\alpha}\partial_{\alpha} + \sum_{\{\alpha:q_\alpha=1\}} r^\sigma \partial_\sigma.
\end{equation}
We will, for the most part, restrict ourselves in this paper to those Frobenius manifolds with the property that $r^\sigma=0\,,\sigma=1\,,\ldots\,,N\,.$
This excludes those examples coming from quantum cohomology and the extended affine Weyl orbit spaces \cite{DZ}: the application
of the inversion symmetry (defined below) even to the simplest prepotential
\[
F=\frac{1}{2} (t^1)^2 t^2 + e^{t^2}
\]
leads to prepotentials of non-analytic functions (e.g. functions such as $e^{-1/x}$ will appear).
\subsection{The deformed connection}\label{gmeq}
One may define on a Frobenius manifold a one parameter family of flat connections parameterized by $z \in \PP^1$,
\begin{equation}
{}^\eta {\widetilde\nabla}_X Y = {}^\eta\nabla_X Y -z X \circ Y, \quad X, Y, Z \in \Gamma(T\cM).
\end{equation}
Here ${}^\eta\nabla$ is the Levi-Civita connection of the metric $\eta\,.$ This connection is flat identically in $z$ by virtue of the
axioms of a Frobenius manifolds. Thus there exists a choice of
coordinates\footnote{The notation $\mathfrak{t}$ will be used rather than the more standard $\tilde{t}$ to avoid, when we look action on $I$,
symbols of the form $\hat{\tilde{t}}\,.$}
$\mathfrak{t}_{\alpha}(t, \lambda)$ such that
\begin{equation}
{}^\eta{\widetilde\nabla} d\mathfrak{t}_{\alpha} = 0.
\end{equation}
In the flat coordinates $(t^1, ... , t^N)$ of connection ${}^\eta\nabla$ this reads
\begin{equation}\label{gm1}
\frac{\partial^2 \mathfrak{t}^{\alpha}}{\partial t^{\kappa}\partial t^{\varepsilon}} = z c^{\sigma}_{\kappa\varepsilon}\frac{\partial \mathfrak{t}^{\alpha}}{\partial t^{\sigma}}.
\end{equation}
Taking in to account the fact that for $z=0$ the two coordinate systems $(t^1, ... t^N)$ and $(\mathfrak{t}^1, ... , \mathfrak{t}^N)$ coincide, one may construct power series solutions to \eqref{gm1}:
\begin{equation*}
\mathfrak{t}^{\alpha}({\bf t},z) = \sum_{n=0}^{\infty} z^n h^{(n,\alpha)}({\bf t}).
\end{equation*}
This ansatz yields a recursion relation for the functions $h^{(n,\alpha)}$,
\begin{equation}\label{recur1}
\frac{\partial^2 h^{(n,\alpha)}}{\partial t^{\kappa}\partial t^{\varepsilon}} =  c^{\sigma}_{\kappa\varepsilon}\frac{\partial h^{(n-1,\alpha)}}{\partial t^{\sigma}}, \quad h^{(0,\alpha)} := t_{\alpha} = \eta_{\alpha\varepsilon}t^{\varepsilon}, \quad n\geq 0.
\end{equation}
Thus starting from the Casimirs $t_\alpha$ one may construct the corresponding deformed flat coordinates recursively.
\medskip
\subsection{From Frobenius Manifolds to Equations of Hydrodynamic Type}
Let $\cM$ be a Frobenius manifold with local coordinates $(u^1, ... , u^N)$. We define on the the loop space of $\cM$ a Poisson bracket of hydrodynamic type
\begin{equation}\label{hydrobrkt1}
\{ H_{(n,\varepsilon)} , H_{(m,\nu)} \} = \int_{S^1} \frac{\delta H_{(n,\varepsilon)}}{\delta u^\alpha} \left( \eta^{\alpha\beta}\frac{\partial~}{\partial X} -\eta^{\alpha\sigma}\Gamma_{\sigma\kappa}^\beta \, t^{\kappa}_{,X} \right) 
\frac{\delta H_{(m,\nu)}}{\delta u^{\beta}} dX,
\end{equation}
where the functional densities for $H_{(n,\varepsilon)}$ depend on $u$ and not its derivatives. Indeed, we define
\begin{equation}\label{conserv1}
H_{(n,\varepsilon)}= \int_{S^1} h^{ (n,\varepsilon) }(u)dX.
\end{equation}
The bracket \eqref{hydrobrkt1} may be used to define the so-called principal hierarchy of the Frobenius manifold
\begin{equation}
\frac{\partial u^{\alpha}}{\partial T_{(n,\kappa)}} = \{ u^{\alpha} , H_{(n,\kappa)} \}.
\end{equation}
In the flat coordinate system $(t^1, ... , t^N)$ this becomes
\begin{equation}
\frac{\partial t^{\alpha}}{\partial T_{(n,\kappa)}}  = \eta^{\alpha\beta}\frac{\partial~}{\partial X}\left( \frac{\partial h^{(n,\kappa)}}{\partial t^{\beta}} \right) = \eta^{\alpha\beta}\frac{\partial^2 h^{(n,\kappa)}}{\partial t^{\beta}\partial t^{\nu}}\frac{\partial t^{\nu}}{\partial X}.
\end{equation}
One may use the relation \eqref{recur1} to illuminate further how the properties of these equations depend precisely on the algebraic structure on the tangent spaces $T_t\cM$ viz:
\begin{equation}\label{matrixform1}
\frac{\partial t^{\alpha}}{\partial T_{(n,\kappa)}} =  \eta^{\alpha\beta}c^{\sigma}_{\beta\lambda}\frac{\partial h^{(n-1,\kappa)}}{\partial t^{\sigma}} \frac{\partial t^{\lambda}}{\partial X}.
\end{equation}
For future use we define
\[
M_{(n,\kappa)}({\bf t})^{\alpha}_{\hspace{0.1cm}{\lambda}}=\eta^{\alpha\beta}c^{\sigma}_{\beta\lambda}\frac{\partial h^{(n-1,\kappa)}}{\partial t^{\sigma}} \,.
\]
and note these flow equations simplify when $\alpha=N\,,$ i.e.
\begin{equation}
\frac{\partial t^N}{\partial T_{(n,\kappa)}} = \frac{\partial h^{(n-1,\kappa)}}{\partial X}\,.
\label{PT}
\end{equation}
Since $t_1=t^N$ is the Egorov potential of the metric $\eta\,,$ this shows that the potential is a conserved quantity of all the flows,
in agreement with \cite{PT} in the more general setting of semi-Hamiltonian Egorov systems.
 \subsection{Symmetries of the WDVV Equations}

The space of solutions of the WDVV equation possesses certain natural symmetries.

\begin{Def} A \emph{symmetry} of WDVV is a map
\begin{equation}
F(t) \mapsto \hat{F}(\hat{t}), \quad t^{\alpha} \mapsto \hat{t}^{\alpha}, \quad \eta_{\alpha\beta} \mapsto \hat{\eta}_{\alpha\beta}
\end{equation}
from one solution to another.
\end{Def}
In \cite{dubrovin1} two symmetries were defined:

\begin{Def}
{}
\begin{itemize}

\item[(a)]\emph{Legendre-type transformations $S_\kappa$}: This is defined by
\begin{eqnarray*}
{\hat t}_\alpha & = & \partial_{t^\alpha} \partial_{t^\kappa} F(t)\,,
\qquad ({\rm N.B.}\quad{\hat t}_\alpha=\eta_{\alpha\beta}t^\beta)\\
\frac{\partial^2{\hat F}}{\partial{\hat t}^\alpha \partial{\hat t}^\beta} & = &
\frac{\partial^2{     F}}{\partial{     t}^\alpha \partial{     t}^\beta} \,,\\
{\hat\eta}_{\alpha\beta} & = & \eta_{\alpha\beta}\,.
\end{eqnarray*}

\medskip

\item[(b)]\emph{Inversion Symmetry:} This is defined by
\begin{displaymath}
\hat{t}^1  = \frac{1}{2}\frac{t_{\sigma}t^{\sigma}}{t^N}, \quad \hat{t}^{\alpha} = \frac{t^{\alpha}}{t^N}, \quad\mbox{  (for }\alpha \neq 1, N), \quad \hat{t}^N = -\frac{1}{t^N},
\end{displaymath}
\begin{equation}\label{isdef}
\hat{F}(\hat{t}) = (\hat{t}^N)^2 F \left( \frac{1}{2}\frac{\hat{t}_{\sigma}\hat{t}^{\sigma}}{\hat{t}^N}, -\frac{\hat{t}^{2}}{\hat{t}^N}, ... , -\frac{\hat{t}^{N-1}}{\hat{t}^N},   -\frac{1}{\hat{t}^N} \right) + \frac{1}{2}\hat{t}^1\hat{t}_{\sigma}\hat{t}^{\sigma},
\end{equation}
\begin{displaymath}
\hat{\eta}_{\alpha\beta} = \eta_{\alpha\beta}.
\end{displaymath}
\end{itemize}
\end{Def}

\medskip

\noindent In this paper we will concentrate on inversion symmetries (also called type II symmetries). For these,
the
structure constants of the inverted Frobenius manifold are related to those of the original by
\begin{equation}\label{ctrans1}
\hat{c}_{\alpha\beta\gamma} = (t^N)^{-2}\frac{\partial t^{\lambda}}{\partial \hat{t}^{\alpha}}\frac{\partial t^{\mu}}{\partial \hat{t}^{\beta}}\frac{\partial t^{\nu}}{\partial \hat{t}^{\gamma}}c_{\lambda\mu\nu}.
\end{equation}

\noindent Under the assumption that 
that $r_i=0$ for $i=1\,,\ldots\,,N$
the corresponding Euler vector field of the inverted Frobenius manifold $\hat{F}$ have the form
\begin{equation}
\hat{E}(\hat{t}) = \sum_{\alpha}(1-\frac{\hat{d}}{2}-\hat{\mu}_{\alpha})\hat{t}^{\alpha}\hat{\partial}_{\alpha},
\end{equation}
where
\begin{equation}
\hat{d}=2-d, \quad \hat{\mu}_1 = \mu_N - 1, \quad \hat{\mu}_N = \mu_1 + 1, \quad \hat{\mu}_i = \mu_i, \quad \mbox{for }i \neq 1, N.
\end{equation}

\begin{Def} A Frobenius manifold is said to be \emph{modular} if it lies at a fixed point of the inversion symmetry,
\begin{equation}
\hat{d} = d\,,\qquad \hat{\mu}_i = \mu_i\,.
\end{equation}
\end{Def}
Comparison of the two Euler fields gives the following:
\begin{prop}
A Frobenius manifold is modular if and only if $d=1$ and $r^N=0.$
\end{prop}
Examples of such modular Frobenius manifolds were given in the introduction. The application of $I$ to such
a modular manifold does not yield a new prepotential but rather a modular tranformation of itself, i.e.
\begin{equation}
F(\hat{t}) = \left({\hat{t}}^N\right)^2 F\left(t({\hat{t}})\right) + \frac{1}{2} \hat{t}^1 {\hat{t}}_\sigma {\hat{t}}^\sigma\,.
\label{invariantF}
\end{equation}
\begin{example}\label{basicexample}
The prepotential
\[
F=\frac{1}{2} u^2 \tau - \frac{1}{2} u ({\bf z},{\bf z}) + f({\bf z},\tau)
\]
(where $({\bf x},{\bf y})$ is some inner product) satisfies \eqref{invariantF} under the action of $I$ if and only if
\[
f\left(\frac{{\bf z}}{\tau},-\frac{1}{\tau}\right) = \frac{1}{\tau^2} f({\bf z},\tau) - \frac{1}{4 \tau^3} ({\bf z},{\bf z})^2\,.
\]
In one dimension the ansatz $f(z,\tau)=z^4 \gamma(\tau)$ reduced this condition to
\[
f\left(-\frac{1}{\tau}\right) = \frac{1}{\tau^2} f(\tau) - \frac{1}{4} \tau\,.
\]
This is the transformation property of the solution to Chazy's equation. Note that the WDVV equations still have to
be satisfied - this argument just gives the modular transformation properties that the solution must satisfy.
\end{example}

\textbf{Remark}. There exists another choice of coordinates on a semi-simple Frobenius manifold, namely one that simplifies its algebraic structure. In these coordinates that multiplication is trivial:
\[
\frac{\partial~}{\partial u^i} \circ \frac{\partial~}{\partial u^j} = \delta_{ij} \frac{\partial~}{\partial u^i}\,.
\]
It turns out that these coordinates are the roots of
\begin{equation}\label{canondef}
\det\left(g^{\alpha\beta}(t)-u\,\eta^{\alpha\beta}\right)=0\,,
\end{equation}
where $g$ is the intersection form of the Frobenius manifold.  The roots of the expression \eqref{canondef} are invariant under the symmetry $I$, and so the canonical coordinates are preserved up to a re-ordering \cite{dubrovin1}.

\section{Inversion and Almost Duality}

Consider the vector field $E^{-1}$ defined by the condition
\[
E^{-1} \circ E = e\,.
\]
This is defined on $M^{\star}= M\backslash \Sigma\,,$ where $\Sigma$ is the discriminant submanifold
on which $E^{-1}$ is undefined. With this field one may define a new \lq dual\rq~multiplication
$\star: TM^\star \times TM^\star \rightarrow TM^\star$ by
\[
X \star Y = E^{-1} \circ X \circ Y\,, \qquad\qquad \forall\, X\,,Y \in TM^\star\,.
\]
This new multiplication is clearly commutative and associative, with the Euler vector field being the
unity field for the new multiplication.

Furthermore, this new multiplication is compatible with the intersection form $g$ on the Frobenius manifold,
i.e.
\[
g(X\star Y, Z) = g(X,Y\star Z)\,, \qquad\qquad \forall\, X\,,Y\,,Z \in TM^\star\,.
\]
Here $g$ is defined by the equation
\[
g(X,Y)=\eta(X\circ Y, E^{-1})\,, \qquad\qquad \forall\, X\,,Y \in TM^\star
\]
(and hence is well-defined on $M^\star\,$). Alternatively one may use the metric $\eta$ to extend the
original multiplication to the cotangent bundle and define
\[
g^{-1}(x,y) = \iota_E(x\circ y) \,, \qquad\qquad \forall\, x\,,y \in T^\star M^\star\,.
\]
The intersection form has the important property that it is flat, and hence there exists a distinguished
coordinate system $\{ p^i\,,i=1\,,\ldots\,,N\}$ in which the components of the intersection form are constant. 
It may be shown that there exists a dual prepotential $F^\star$ such that its third derivatives give the structure
functions $c^{\star}_{ijk}$ for the dual multiplication. More precisely \cite{dubrovin2}:

\begin{theorem} Given a Frobenius manifold $M$, there exists a function $F^\star$ defined on $M^\star$
such that:
\begin{eqnarray*}
c^{\star}_{ijk} & = &
g\left( \frac{\partial~}{\partial p^i}\star \frac{\partial~}{\partial p^j}\,, \frac{\partial~}{\partial p^k}
\right) \,,\\
& = &\frac{\partial^3 F^\star}{\partial p^i \partial p^j \partial p^k}\,.
\end{eqnarray*}
Moreover, the pair $(F^\star,g)$ satisfy the WDVV-equations in the flat coordinates $\{ p^i \}$ of the metric $g\,.$
\end{theorem}

Thus given a specific Frobenius manifold one may construct a \lq dual\rq~solution
to the WDVV-equations by constructing the flat-coordinates of the intersection
form and using the above result to find the tensor $c^{\star}_{ijk}$ from which the
dual prepotential may be constructed.

Recall the following:

\begin{Def}
A function $\mathrm{p}=\mathrm{p}(t;\lambda)$ is called a $\lambda$-period of the Frobenius manifolds
if it satisfies the Gauss-Manin equations
\[
\left({}^g \nabla \,\,-\lambda \,\,{}^{\eta} \nabla \right) d\mathrm{p}=0\,.
\]
\end{Def}
\noindent The following fact were proved in \cite{dubrovin2}:
\begin{itemize}
\item[$\bullet$] $p=\mathrm{p}(t;0)$ is a flat coordinate for the intersection form $g\,;$
\item[$\bullet$] $\mathrm{p}(t;\lambda) = p(t^1-\lambda\,,t^2\,,\ldots\,,t^N)\,.$
\end{itemize}
To understand the action of $I$ in the dual picture it is first necessary to understand the relationship
between the two $\lambda$ periods $\hat{\mathrm{p}}$ and $\mathrm{p}$, or just the relationship
between the flat coordinates of $g$ and $\hat{g}\,.$

\begin{lemma}\label{basiclemma}  Let $p$ be a flat coordinate for $g$ with the property $E(p)=\left(\frac{1-d}{2}\right)p\,.$
Then
\[
{\hat p}=\frac{p}{t_1}
\]
is a flat coordinate for ${\hat g}\,.$ Moreover
\[
{\hat{\mathrm{p}}}=\frac{\mathrm{p}}{t_1}
\]
\end{lemma}

\begin{proof}
Recall that a flat coordinate $p$ for the metric $g$ must satisfy the equation ${}^g\nabla dp=0\,.$ Thus one
just has to calculate ${}^{\hat{g}}\nabla d{\hat{p}}.$ In canonical coordinates
\[
\eta=\sum_{i=1}^N H_i^2({\bf u}) du_i^2\,, \qquad g=\sum_{i=1}^N \frac{H_i^2({\bf u})}{u_i} du_i^2
\]
with Egorov potential $t_1\,,$ i.e. $H_i^2=\partial_i t_1\,.$ In terms of the dual objects, ${\hat{u}}_i=u_i$ and
\[
{\hat{H}}_i=\frac{H_i}{t_1}\,, \qquad \hat{t}_1=-\frac{1}{t_1}
\]
and the rotation coefficients satisfy the relation ${\hat{\beta}}_{ij}=\beta_{ij} - \frac{H_i H_j}{t_1}\,.$
Hence the connections ${}^{\hat{g}}\nabla$ and ${}^{{g}}\nabla$  can be expressed in terms of the Darboux-Egorov data for $\eta$ and the
result follows via straightforward calculations.
\end{proof}
We now move from a flat coordinate to flat coordinate systems. It is here that the differences between the cases
$d\neq 1$ and $d=1$ becomes apparent.

\begin{prop}\label{dneqone} Let $\{p^i\,:i=1\,,\ldots\,,N\}$ be a flat coordinate system for $g\,.$ If $d\neq 1$ then
\[
\left\{{\hat{p}}^i=\frac{p^i}{t_1} \,:i=1\,,\ldots\,,N\right\}
\]
is a flat coordinate system for ${\hat{g}}\,.$ Moreover the normalization conditions are preserved, and the Gram
matrices coincide.
\end{prop}
\begin{proof}
In \cite{dubrovin2} it was shown that if $d\neq 1$ then $p$ may be normalized so $E(p)=\left(\frac{d-1}{2}\right)p\,,$ so the above Lemma
shows that the ${\hat{p}}^i$ are flat. To show that they form a coordinate {\em system} one must calculate the Jacobian
of the transformation. A simple calculation gives
\[
\frac{\partial({\hat{p}^1}\,,\ldots\,,{\hat{p}}^N)}{\partial(p^1\,,\ldots\,,p^N)} = - \frac{1}{t_1^{N+1}}
\]
on using the result \cite{dubrovin2} (again for $d\neq 1$) that $t_1=\left( \frac{1-d}{2} \right) g_{ab}p^a p^b\,.$
The calculation of the dual normalization conditions and Gram matrices is also straightforward and yields the relation ${\hat{g}}_{ab}=g_{ab}\,.$
\end{proof}

\noindent When $d=1$ certain special results hold.

\begin{lemma} Suppose $d=1\,.$ Then:
\begin{itemize}
\item[(i)] ${}^g \nabla E=0\,;$
\item[(ii)] $t_1$ is a flat coordinate for both $\eta$ and $g\,.$
\end{itemize}
\end{lemma}

\begin{proof} The proof is by direct computation. For example, in flat coordinates $\{ t^i\,,i=1\,,\ldots\,,N\}$ for $\eta$ the Christoffel
symbols for $g$ take the form \cite{dubrovin1,dubrovin2}
\[
{}^g \Gamma^{\alpha\beta}_\gamma = c^{\alpha\epsilon}_\gamma \left( \frac{1}{2} - \mathcal{V}\right)^\beta_\epsilon
\]
and using this one finds that ${}^g\nabla_\alpha E^\beta=\left(\frac{1-d}{2}\right) \delta^\beta_\alpha\,.$ Alternatively, in
canonical coordinates,
\begin{eqnarray*}
{}^g \nabla_i E^i & = & 0 \,, \qquad i\neq j\,,\\
{}^g \nabla_i E^i & = & \frac{(1-d)}{2} \,,\qquad ({\rm no~sum})\,.
\end{eqnarray*}
The proof that $t_1$ is flat for $g$ (i.e. ${}^g\nabla dt_1=0$) is similar.
\end{proof}

\noindent Recall that $E$ plays the role of the unity vector field in the dual picture. Thus when $d=1$ the unity vector
field is covariantly constant, so almost dual Frobenius manifolds at $d=1$ are even closer to Frobenius manifolds than for those
with $d\neq1\,.$ We will return to this point later.

\begin{theorem} Suppose $d=1$ and $r_N=0$ and let $\{p^i\,:i=1\,,\ldots\,,N\}$ be a flat coordinate system for $g\,.$
Let
\begin{displaymath}
\hat{p}^1  = \frac{1}{2}\frac{p_{\sigma}p^{\sigma}}{p^N}, \quad \hat{p}^{\alpha} = \frac{p^{\alpha}}{p^N}, \quad\mbox{  (for }\alpha \neq 1, N), \quad \hat{p}^N = -\frac{1}{p^N},
\end{displaymath}
where $p^N=t_1\,.$ Then
$\{{\hat{p}}^i\,:i=1\,,\ldots\,,N\}$ are a flat coordinate system for ${\hat{g}}\,.$
\end{theorem}

\begin{proof} From the above lemma $t_1$ is a flat coordinate for $g$ and hence we choose $p^N=t_1\,.$ With this $E(p^N)=0$ and since
${}^g\nabla E=0\,,$ the vector field $E$ must take the form $E=\sum_{i=1}^{N-1} c_i \frac{\partial~}{\partial p^i}$ for some constants $c_i\,.$
Using the freedom to redefine the $p^i$ for $i\neq N$ one may set
\[
E=\frac{\partial~}{\partial p^1}\,.
\]
With this $g(E,E)=\eta(E^{-1} \circ E,E) = \eta(e,E)=r_N$ (again since $d=1$). Thus from \cite{dubrovin1} Lemma 1.1 one may
redefine coordinates so
\[
g_{ab}=\delta_{a+b,N+1}
\]
in the $\{p^i\}$-coordinates. Since $E(p^i)=0.p^i$ for $i=2\,,\ldots\,,N-1$ one may use Lemma \ref{basiclemma} to show that
${\hat{p}}^i=p^i/p^N$ are flow coordinates. By direct calculation - for example, in canonical coordinates - one may show that ${\hat{p}}^1$ and ${\hat{p}}^N$ define
above are also flat coordinates (again, one just has to show ${}^{\hat{g}} \nabla d{\hat p}=0\,$). It also follows immediately that
${\hat{g}}_{ab}=g_{ab}\,.$
\end{proof}
It remains to compare the two induced prepotentials $F^\star$ and ${\hat{F^\star}}$ and the corresponding multiplications.
\begin{theorem}
Let $F$ define a Frobenius manifold and let $\hat{F}$ denote the induced manifold under the action of the symmetry $I\,.$ Let $F^\star$ and
$\hat{F^\star}$ denote the corresponding almost dual structures. The $I^\star$, the induced symmetry act as:

\begin{itemize}

\item{Case I: $d \neq 1\,:$}

\begin{eqnarray*}
{\hat p}^i & = & \frac{p^i}{t_1} \,, \qquad i=1\,,\ldots\,, N\,,\\
{\hat{g}}_{ab} & = & g_{ab} \,, \\
{\hat{F^\star}}({\hat{\mathbf{p}}}) & = & 
\frac{F^\star(\mathbf{p}({\hat{\mathbf{p}}}))}{t_1^2}
\end{eqnarray*}
where $t_1=\frac{(1-d)}{2} g_{ab} p^a p^b\,.$

\medskip

\item{Case II: $d=1\,:$}

\begin{eqnarray*}
{\hat p}^1 & = & \frac{1}{2}\frac{p_{\sigma}p^{\sigma}}{t_1}\,,\quad
{\hat p}^i  =  \frac{p^i}{t_1} \,, \qquad i=2\,,\ldots\,, N-1\,,\quad
{\hat p}^N  =  -\frac{1}{t_1}\,,\\
{\hat{g}}_{ab} & = & g_{ab} \,, \\
{\hat{F^\star}}(\hat{\mathbf{p}}) & = & (\hat{p}^N)^2 F \left( \mathbf{p}({\hat{\mathbf{p}}}) \right) + \frac{1}{2}\hat{p}^1\hat{p}_{\sigma}\hat{p}^{\sigma}\,,
\end{eqnarray*}

\noindent where $t_1=p^N\,.$

\medskip

\end{itemize}
\noindent Note, in both cases $p^N$ is the Egorov potential for the metric $\eta\,.$

\end{theorem}

\begin{proof}

Recall \cite{dubrovin2} that for $d\neq 1$ the dual prepotential satisfies the homogeneity condition
\[
\sum_\alpha p^\alpha \frac{\partial F^\star}{\partial p^\alpha} = 2 F^\star + \frac{1}{1-d} g_{\alpha\beta} p^\alpha p^\beta\,.
\]
Using this and the explicit coordinates given in Proposition \ref{dneqone} one finds that
\[
\frac{\partial^3 {\hat{F^\star}}}{\partial{\hat{p}}^\alpha \partial{\hat{p}}^\beta \partial{\hat{p}}^\gamma} =
t_1 \frac{\partial^3 {{F^\star}}}{\partial{{p}}^\alpha \partial{{p}}^\beta \partial{{p}}^\gamma}
-\frac{2}{1-d} (g_{\alpha\beta} p_\gamma+g_{\alpha\gamma} p_\beta+g_{\beta\gamma} p_\alpha) + \frac{2}{t_1(1-d)} p_\alpha p_\beta p_\gamma\,.
\]
From this it is straightforward to show that $\hat{F^\star}$ satisfies the WDVV equations in the $\{ p^i\}$-variables.

\medskip

If $d=1$ the proof is identical to the original inversion symmetry as presented in \cite{dubrovin1}.

\end{proof}

It remains to show how the deformed flat coordinates of both ${}^\eta {\widetilde{\nabla}}$ and ${}^g {\widetilde{\nabla}}$
behave under inversion symmetry. This will be done in the next section. We will end this section with some examples.

\begin{example}[$d\neq 1$] Given an irreducible Coxeter group $W$ of rank $N\,,$ the Saito construction gives a Frobenius manifold
structure on the orbit space $\mathbb{C}^N/W\,.$ The almost dual prepotential takes the form
\[
F^\star({\bf p}) = \frac{1}{4} \sum_{\alpha\in R_W} (\alpha,{\bf p}) \log (\alpha,{\bf p})^2
\]
where $(,)$ is the metric $g\,.$

Application of the $I^\star$ transform (recall $d\neq 1$ for these examples) yields the solution
\[
{\hat{F^\star}}({\hat{\bf p}}) = \frac{1}{4} \sum_{\alpha\in R_W} (\alpha,{\bf {\hat{p}}}) \log (\alpha,{\hat{{\bf p}}})^2 - \frac{h}{4} ({\hat{{\bf p}}},{\hat{{\bf p}}})\log({\hat{{\bf p}}},{\hat{{\bf p}}})
\]
where $h$ is defined by the relation $\sum_{\alpha\in R_W} (\alpha,{\bf z}) = h ({\bf h},{\bf h})$ (and hence depends on the normalization of the
roots $\alpha\in R_W\,).$

Thus the original solution is recovered but with the addition of a new radial term. Such solution have been constructed directly (i.e. without
knowledge of its geometric origins) in \cite{LP}. The transformation property hold more generally than just for dual solutions, and hence may also be applied to the $\vee$-systems \cite{Veselov}.

\end{example}

\begin{example}[$d=1$] Given the Weyl group $A_N$ and $B_n$ with Lie algebra $\mathfrak{g}$ with Cartan subalgebra $\mathfrak{h}$
one may construct the so-called Jacobi group $J(\mathfrak{h})$ and orbit space $\Omega/J(\mathfrak{g})$ where
$\Omega=\mathbb{C}\oplus {\mathfrak h} \oplus \mathbb{H}\,.$ This orbit space carries the structure of a Frobenius
manifold and the dual prepotential takes the form

\begin{eqnarray*}
F^\star(u\,,{\bf z}\,,\tau) & = &
\frac{1}{2} \tau u^2 -
\frac{1}{2} u ({\bf z},{\bf z})+\sum_{\alpha\in\mathfrak{U}} h(\alpha.{\bf z},\tau)
\end{eqnarray*} Here the function $h$ is essentially the
elliptic trilogarithm introduced by Beilinson and Levin \cite{BL,Levin} and the set
$\mathfrak{U}$ contains certain vectors - an elliptic generalization of classical root
systems. The basic function $h$ satisfies the modularity property (c.f. example \ref{basicexample})
\[
h\left( \frac{z}{\tau} \right) = \frac{1}{\tau^2} h(z,\tau) - \frac{z^4}{4!\,\tau^3}
\]
up to quadratic terms. The proof that this satisfies the inversion symmetry may be found in \cite{iabs2}.

\end{example}

\section{Inversion Symmetry and Principal Hierarchies}\label{mainresult}

A deformed flat coordinate $\mathfrak{t}$ must satisfy the equation ${}^\eta \widetilde{\nabla} d\mathfrak{t}=0$ and similarly
a deformed flat coordinate $\mathfrak{p}\,,$ also known as twisted period, must satisfy the equation ${}^g \widetilde{\nabla} d\mathfrak{p}=0\,,$ where
${}^g {\widetilde\nabla}_X Y = {}^g \nabla_X Y - \nu X \star Y$ is the dual deformed connection. It is straightforward to prove,
using canonical coordinates, that

\[
\hat{\mathfrak{t}}=\frac{\mathfrak{t}}{t_1} \qquad {\rm and~}\qquad {\hat{\mathfrak{p}}}=\frac{\mathfrak{p}}{t_1}
\]
are corresponding dual deformed flat coordinates (and one follows from the other by the contour integral/Laplace transform
methods in \cite{dubrovin2}). However, a more subtle form appears when the calculations are performed in flat coordinates -
a shift appears in the labels. We begin by examining the relationship between $\mathfrak{t}$ and $\hat{\mathfrak{t}}$ since the
coefficients will give the transformation between the Hamiltonian densities of the various flows. These results has also been proved independently in \cite{zhang}.

\begin{prop}\label{propzero} Under the inversion symmetry, the Hamiltonian densities of the principal hierarchy corresponding to $F$ are mapped to those of $\hat{F}$ according to the following rules
\begin{equation}\label{htrans1}
h^{(n,\alpha)} (\mathbf{t}(\hat{\mathbf{t}})) = -\frac{1}{\hat{t}^N}\hat{h}^{(\tilde{n},\tilde{\alpha})}(\hat{\mathbf{t}}),
\end{equation}
where
\begin{equation*}
\tilde{n} = \left\{ \begin{array}{ll} n+1, & \mbox{if } \alpha = N, \\ n, & \mbox{if }\alpha \neq 1, N, \\ n-1, &\mbox{if }\alpha = 1, \end{array}\right.
\tilde{\alpha} = \left\{ \begin{array}{ll} 1, & \mbox{if } \alpha = N, \\ \alpha, & \mbox{if }\alpha \neq 1, N, \\ N, & \mbox{if }\alpha = 1. \end{array}\right.
\end{equation*}
\end{prop}
\begin{proof} This consists of inverting equation \eqref{gm1} using \eqref{ctrans1} and checking the solutions given in the proposition \eqref{htrans1} do in fact satisfy it. One may identify how the labels $\alpha$ are mapped using the fact that the unity field decreases the degree of the densities by one,
\begin{equation*}
e^{\tilde{n}}(h^{(\tilde{n},\tilde{\alpha})}) = t_{\tilde{\alpha}}.
\end{equation*}
Here $e^{\tilde{n}}$ denotes $\tilde{n}$ successive applications of the operator $e= \partial / \partial t^1$.
\end{proof}
This lifts in turn to the corresponding flows, leading to the following
\begin{prop}\label{prop1} Under the inversion symmetry, the flows on the loop space of the inverted Frobenius manifold are related to those of the original one via
\begin{equation}\label{flowtrans1}
M_{(n,\alpha)}(\mathbf{t}(\hat{\mathbf{t}}))= -\hat{t}^N\hat{M}_{(\tilde{n},\tilde{\alpha})}(\hat{\mathbf{t}}) + \hat{h}^{(\tilde{n}-1,\tilde{\alpha})}(\hat{\mathbf{t}})\mathbf{1},
\end{equation}
where $\tilde{n}$, $\tilde{\alpha}$ are as above.
\end{prop}
\begin{proof} This is a calculation that follows from plugging the results of Proposition \ref{propzero} into \eqref{matrixform1}. Particularly, using equations \eqref{isdef}, \eqref{ctrans1}, and \eqref{htrans1} we can apply $I$ to \eqref{matrixform1} to find
\setlength \multlinegap{0pt}
\begin{multline} \tag{2}
\frac{\partial\hht^{\rho}}{\partial T_{(n,\alpha)}} = -\underbrace{(\hht^N)^{-3} \frac{\p \hht^{\rho}}{\p t^{\iota}} \hat{\eta}^{\iota\kappa} \hat{\eta}^{\mu\lambda} \frac{\p \hht^{\delta}}{\p t^{\lambda}} \frac{\p \hht^{\gamma}}{\p t^{\kappa}} \hat{c}_{\delta\gamma\xi}(\hht)\frac{\p \hht^{\ve}}{\p t^{\mu}} \frac{\p \hh^{(\tilde{n}-1,\tilde{\alpha})}(\hht)}{\p t^{\mu}}\frac{\partial \hht^{\xi}}{\partial X}}_{\mbox{term }1} \\
+ \underbrace{(\hht^N)^{-2} \frac{\p \hht^{\rho}}{\p t^{\iota}} \hat{\eta}^{\iota\kappa} \frac{\p \hht^{\gamma}}{\p t^{\kappa}} \hat{\eta}_{\gamma\xi}\hh^{(\tilde{n}-1,\tilde{\alpha})}(\hht)\frac{\partial \hht^{\xi}}{\partial X}}_{\mbox{term }2}.
\end{multline}
A straightforward but tedious calculation then gives simplification of the terms
\begin{equation*}
\mbox{term }1 = -\hht^N \hat{\eta}^{\rho\ve}\hat{c}^{\mu}_{\ve\sigma}\frac{\p \hh^{(n-1,\alpha)}(\hht)}{\p \hht^{\mu}}\frac{\partial \hht^{\sigma}}{\partial X} \qquad \mbox{term }2= \delta^{\rho}_{\sigma}\hh^{(n-1,\alpha)}(\hht)\frac{\partial \hht^{\sigma}}{\partial X}.
\end{equation*}
Proposition is proved.
\end{proof}

Naturally, one is led to consider the analogous problem in the almost dual picture. This gives the following

\begin{prop}
Assume $d \neq 1$. Let $F^\star$ and $\hat{F}^\star$ be related by the almost dual inversion symmetry, $I^\star$. Then the successive approximations to the flat coordinates for the corresponding deformed connections $^g\tilde{\nabla}$ and $^{\hat{g}}\tilde{\nabla}$ are related via
\begin{equation}\label{ptrans1}
l^{(n,\alpha)}(\mathbf{p}(\hat{\mathbf{p}})) = \frac{1}{\hat{t}_1}\hat{l}^{(n,\alpha)}(\hat{\mathbf{p}}),
\end{equation}
where the approximations $l^{(n,\alpha)}$ are defined via
\begin{equation}
\mathfrak{p}^{\alpha}(\mathbf{p},\nu) = \sum_{n=0}^{\infty} \nu^n l^{(n,\alpha)}(\mathbf{p})
\end{equation}
This gives rise to the corresponding relationship between the almost dual flows
\begin{equation}\label{adflowtrans1}
M_{(n,\alpha)}(\mathbf{p}(\hat{\mathbf{p}})) = \hat{t}_1\hat{M}_{(n,\alpha)}(\hat{\mathbf{p}}) - \frac{2}{1-d}\hat{l}^{(n-1,\alpha)}(\hat{\mathbf{p}})\mathbf{1}.
\end{equation}
If $d=1$ the transformation is the same as in \eqref{htrans1}.
\end{prop}

\begin{proof}
This is carried out in an analogous manner to the above using
\begin{equation*}
\frac{\partial^3 {\hat{F^\star}}}{\partial{\hat{p}}^\alpha \partial{\hat{p}}^\beta \partial{\hat{p}}^\gamma} =
t_1 \frac{\partial^3 {{F^\star}}}{\partial{{p}}^\alpha \partial{{p}}^\beta \partial{{p}}^\gamma}
-\frac{2}{1-d} (g_{\alpha\beta}p_\gamma+g_{\alpha\gamma}p_{\beta}+g_{\beta\gamma}p_{\alpha}) + \frac{2}{t_1(1-d)} p_\alpha p_\beta p_\gamma\,.
\end{equation*}
One must also make use of certain homogeneity conditions, for example
\begin{equation*}
p^{\sigma}\frac{\partial^3F^\star}{\partial p^{\sigma} \partial p^{\varepsilon} \partial p^{\kappa}} = \frac{2g_{\kappa\varepsilon}}{1-d}.
\end{equation*}
\end{proof}

At first sight the result is surprising: the flows associated to $F$ are not mapped under $I$ to the flows associated to $\hat{F}\,.$ However
a simple reciprocal transformation
\begin{eqnarray*}
d{\tilde{X}} & = & t^N dX + h^{({\tilde{n}}-1,{\tilde{\kappa}})} dT_{({\tilde{n}},{\tilde{\kappa}})}\,,\\
d{\tilde{T}_{({\tilde{n}},{\tilde{\kappa}})}} & = & dT_{({\tilde{n}},{\tilde{\kappa}})}
\end{eqnarray*}
(recall that the Egorov potential $t^N$ is conserved by (\ref{PT})) transform the flow to that associated with ${\hat{F}}\,.$ Thus:

\begin{theorem}
Up to a simple reciprocal transformation, $I$ maps flows of $F$ to flows of $\hat{F}\,.$ Moreover, if the Frobenius manifold
is modular, $I$ maps the flows to themselves, i.e. the flows are invariant.
\end{theorem}

For semi-simple Frobenius manifolds these flow take diagonal form when written in terms of canonical coordinates (which are
specific examples of Riemann invariants):
\[
\frac{\partial u^i}{\partial T_{(n,k)}} = \lambda^i_{(n,k)}({\bf u}) \frac{\partial u^i}{\partial X}\,.
\]
Since the characteristic velocities are the eigenvalues of the tensor $M_{(n,k)}$ one can easily show that
under inversion
\[
{\hat \lambda}^i_{(n,k)} = t^N \lambda^i_{(n,k)} - h_{(n,k)}\,
\]
and thus by applying a reciprocal transformation the system, when written in Riemann invariant form,
is unchanged. This same result holds in the weaker setting of semi-Hamiltonian systems where the metric is Egorov \cite{PT}.
A class of such examples may be found by restricting flows to certain natural submanifolds of a Frobenius manifold \cite{iabs1,FV}.

We end this section with an extended example which shows how the modularity properties of solutions to the Chazy equation
results in, up to a reciprocal transformation, flows that are invariant under the modular transformations.

\begin{example}
Consider the Frobenius manifold defined by the free energy
\begin{equation}
F = \frac{1}{2}t_1^2t_3+\frac{1}{2}t_1t_2^2 - \frac{t_2^4}{16}\gamma(t_3)\,; \quad E=t_1\frac{\partial}{\partial t_1} + \frac{1}{2}t_2\frac{\partial}{\partial t_2},
\end{equation}
where $\gamma$ is some unknown 1-periodic function. For the duration of this example all coordinates will be written with lowered indices ($t_{\alpha} = t^{\alpha}$). In order for $F$ to satisfy WDVV, $\gamma$ must satisfy Chazy's equation,
\begin{equation*}
\gamma'''(t_3)=6\gamma(t_3)\gamma''(t_3)-9(\gamma'(t_3))^2.
\end{equation*}
The main property of the Chazy equation is an $SL(2,\ZZ)$ invariance:
\begin{eqnarray*}
t_3 & \mapsto & \frac{at_3 + b}{ct_3 + d}, \quad ad-bc =1, \nonumber \\
\gamma(t_3) & \mapsto & (ct_3 + d)^2\gamma(t_3)+2c(ct_3 + d).
\end{eqnarray*}
This in turn allows us to apply the inversion symmetry. Since $d=1$ and $r^N=0$ this example is one of a class of examples of Frobenius manifolds that lie at fixed points of the inversion symmetry, and hence defines a modular Frobenius manifold.
Consider the following solution to equation \ref{gm1},
\begin{equation}
h_{ (0,2) } = t_2.
\end{equation}
Applying the inversion symmetry, we find
\begin{equation}
h_{ (0,2) }(\hat{t}_1+\frac{1}{2}\frac{\hat{t}_2^2}{\hat{t}_3}, -\frac{\hat{t}_2}{\hat{t}_3}, -\frac{1}{\hat{t}_3}) = -\frac{\hat{t}_2}{\hat{t_3}} =  -\frac{1}{\hat{t_3}}h_{(0,2)}(\hat{t}_1,\hat{t}_2,\hat{t}_3),
\end{equation}
since $F$ lies at a fixed point of the inversion symmetry. Similarly, considering
\begin{equation}
h_{ (1,1) } = t_1t_3+\frac{1}{2}t_2^2
\end{equation}
one finds
\begin{equation}\begin{split}
h_{ (1,1) } (\hat{t}_1+\frac{1}{2}\frac{\hat{t}_2^2}{\hat{t}_3}, -\frac{\hat{t}_2}{\hat{t}_3}, -\frac{1}{\hat{t}_3}) = \left( \hat{t}_1+\frac{1}{2}\frac{\hat{t}_2^2}{\hat{t}_3} \right)\left( -\frac{1}{\hat{t}_3} \right)+\frac{1}{2}\left( -\frac{\hat{t}_2}{\hat{t}_3} \right) = -\frac{1}{\hat{t}_3}\hat{t}_1 = -\frac{1}{\hat{t}_3}\hat{h}_{(0,3)}.
\end{split} \end{equation}
%
The following depicts how the densities are mapped under $I$:

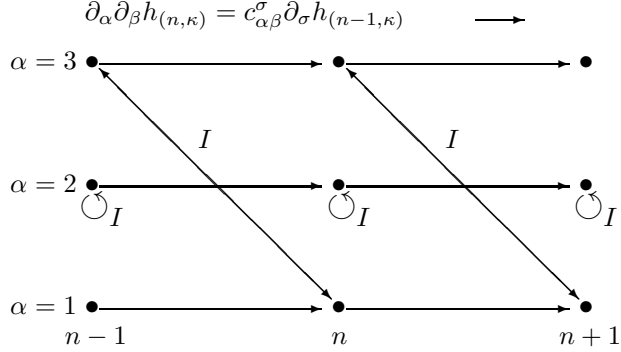
\setlength{\unitlength}{0.65mm}
\begin{center}
\begin{picture}(100,80)\put(-15,50){\small{$\alpha = 3$}}
\put(0,50){\textbullet}
\put(0,60){\small{$\partial_{\alpha}\partial_{\beta}h_{(n,\kappa)}=c^{\sigma}_{\alpha\beta}\partial_{\sigma}h_{(n-1,\kappa)}$}}
\put(79,60){\vector(1,0){10}}
\put(3,51){\vector(1,0){45}}
\put(3,50){\vector(1,-1){47}}
\put(50,3){\vector(-1,1){47}}
\put(50,50){\textbullet}
\put(53,51){\vector(1,0){45}}
\put(53,50){\vector(1,-1){47}}
\put(100,3){\vector(-1,1){47}}
\put(100,50){\textbullet}
\put(23,34){\small{$I$}}
\put(73,34){\small{$I$}}
\put(-15,25){\small{$\alpha = 2$}}
\put(0,25){\textbullet}
\put(-1,20){\Large{\textbf{$\circlearrowleft$}}}
\put(5,18){\small{$I$}}
\put(3,26){\vector(1,0){45}}
\put(50,25){\textbullet}
\put(49,20){\Large{\textbf{$\circlearrowleft$}}}
\put(55,18){\small{$I$}}
\put(53,26){\vector(1,0){45}}
\put(100,25){\textbullet}
\put(99,20){\Large{\textbf{$\circlearrowleft$}}}
\put(105,18){\small{$I$}}
\put(-15,0){\small{$\alpha = 1$}}
\put(0,0){\textbullet}
\put(-4,-6){\small{$n-1$}}
\put(3,1){\vector(1,0){45}}
\put(50,0){\textbullet}
\put(50,-6){\small{$n$}}
\put(53,1){\vector(1,0){45}}
\put(100,0){\textbullet}
\put(96,-6){\small{$n+1$}}
\end{picture}
\end{center}
\vspace{0.5cm}

The flow corresponding the the density $h_{(0,2)}$ is
\begin{equation*}
\frac{\partial}{\partial T_{(1,2)}}\left( \begin{array}{c} t_1 \\ t_2 \\ t_3 \end{array}\right) =
\left( \begin{array}{ccc} 0 & \frac{-3}{4}t_2^4\gamma'(t_3) & -\frac{3}{4}t_2^3\gamma''(t_3) \\ 1 & -\frac{3}{2}\gamma(t_3)t_2 & -\frac{3}{4}t_2^2\gamma'(t_3) \\ 0 & 1 & 0 \end{array}\right)
\frac{\partial}{\partial X}\left( \begin{array}{c} t_1 \\ t_2 \\ t_3 \end{array} \right),
\end{equation*}
Inverting this one finds
\begin{eqnarray}
\frac{\partial}{\partial T_{(1,2)}}\left( \begin{array}{c} \hht_1 \\ \hht_2 \\ \hht_3 \end{array} \right) &=& \left( \begin{array}{ccc} \hht_2 & \frac{3}{4}\hht_2^2\hht_3\gamma'(\hht_3) & \frac{1}{4}\hht_2^3\hht_3\gamma''(\hht_3) \\ -\hht_3 & \hht_2 + \frac{3}{2}\gamma(\hht_3)\hht_2\hht_3 & \frac{3}{4}\hht_2^2\hht_3\gamma'(\hht_3) \\ 0 & -\hht_3 & \hht_2 \end{array} \right)\frac{\partial}{\partial X}\left( \begin{array}{c} \hht_1 \\ \hht_2 \\ \hht_3 \end{array} \right) \nonumber \\
&=&\left( \left( \begin{array}{ccc} 0 & \frac{-3}{4}\hht_2^4\gamma'(\hht_3) & -\frac{3}{4}\hht_2^3\gamma''(\hht_3) \\ 1 & -\frac{3}{2}\gamma(\hht_3)\hht_2 & -\frac{3}{4}\hht_2^2\gamma'(\hht_3) \\ 0 & 1 & 0 \end{array}\right) + \hht_2\mathbf{1} \right)
\frac{\partial}{\partial X}\left( \begin{array}{c} \hht_1 \\ \hht_2 \\ \hht_3 \end{array} \right) \nonumber \\
&=& \left( -\hht_3\hat{M}_{1,2}(\hat{\mathbf{t}}) + h_{(0,2)}(\hat{\mathbf{t}})\mathbf{1}\right) \frac{\partial}{\partial X}\left( \begin{array}{c} \hht_1 \\ \hht_2 \\ \hht_3 \end{array} \right), \nonumber
\end{eqnarray}
as predicted by \eqref{flowtrans1}.

Similarly, one can see how the flows corresponding to $\alpha=1$ and $\alpha=3$ are related. The flow corresponding to $h_{(1,1)}$ is
\begin{equation*}
\frac{\partial}{\partial T_{(2,1)}} \left( \begin{array}{c} t_1 \\ t_2 \\ t_3 \end{array} \right) = \left( \begin{array}{ccc} t_1 & -\frac{3}{4}t_2^3\gamma'(t_3) - \frac{1}{4}t_2^3 t_3 \gamma''(t_3) & -\frac{1}{4}t_2^4\gamma''(t_3)-\frac{1}{16}t_2^4t_3\gamma'''(t_3) \\ t_2 & t_1 - \frac{3}{2}\gamma(t_3)t_2^2-\frac{3}{4}t_2^2t_3\gamma(t_3) & -\frac{3}{4}t_2^3\gamma'(t_3) - \frac{1}{4}t_2^3t_3\gamma''(t_3) \\ t_3 & t_2 & t_1 \end{array} \right)\frac{\partial}{\partial X}\left( \begin{array}{c} t_1 \\ t_2 \\ t_3 \end{array} \right).
\end{equation*}
This may be inverted to give
\begin{eqnarray}
\frac{\partial}{\partial T_{(2,1)}}\left( \begin{array}{c} \hat{t}_1 \\ \hat{t}_2 \\ \hat{t}_3 \end{array} \right) &=& \left( \begin{array}{ccc} \hat{t}_1 & \frac{1}{4}\hat{t}_2^3\hat{t}_3\gamma''(\hat{t}_3) & \frac{1}{16}\hat{t}_2^4\hat{t}_3\gamma'''(\hat{t}_3) \\ 0 & \hat{t}_1 + \frac{3}{4}\hat{t}_2^2\hat{t}_3\gamma'(\hat{t}_3) & \frac{1}{4}\hat{t}_2^3\hat{t}_3\gamma''(\hat{t}_3) \\ -\hat{t}_3 & 0 & \hat{t}_1 \end{array} \right) \frac{\partial}{\partial X} \left( \begin{array}{c} \hat{t}_1 \\ \hat{t}_2 \\ \hat{t}_3 \end{array} \right) \nonumber \\
&=&\left( -\hat{t}_3 \left( \begin{array}{ccc} 0 & -\frac{1}{4}\hat{t}_2^3\hat{t}_2\gamma''(\hat{t}_3) & -\frac{1}{16}\hat{t}_2^4\gamma'''(\hat{t}_3) \\ 0 &  -\frac{3}{4}\hat{t}_2^2 \gamma'(\hat{t}_3) & -\frac{1}{4}\hat{t}_2^3\gamma''(\hat{t}_3) \\ 1 & 0 & 0 \end{array} \right) + \hat{t}_1\mathbf{1} \right) \frac{\partial}{\partial X} \left( \begin{array}{c}  \hat{t}_1 \\ \hat{t}_2 \\ \hat{t}_3 \end{array} \right) \nonumber \\
&=& (-\hat{t}_3\hat{M}_{(1,3)}(\hat{t})+ \hat{h}_{(0,3)}(\hat{t})\mathbf{1})\frac{\partial}{\partial X}\left( \begin{array}{c} \hht_1 \\ \hht_2 \\ \hht_3 \end{array} \right).\nonumber
\end{eqnarray}

\end{example}

Note the use, in the above example, of the modular transformation properties of solutions of Chazy's equation and their derivatives.

\section{Comments}

For Frobenius manifolds with $d=1$ the Euler field - the unity field for the dual multiplication - is covariantly constant
with respect to the Levi-Civita connection of the intersection form $g\,.$ Thus in this case one is even closer to the axioms
of a Frobenius manifolds than are the almost duality axioms. The dual prepotential must take the form
\[
{\hat{F}}({\bf p})={\rm cubic~} + f(p^2\,,\ldots\,,p^N)
\]
where the precise form of the cubic terms depends on the original Frobenius manifold. If $\eta_{11}=0$ then
$g(E,E)=r_N$ and hence depends on whether $r_N$ is zero or not, by Lemma 1.1 in \cite{dubrovin1}. This fact has already been observed in the
explicit calculation of dual prepotentials of various Hurwitz space Frobenius manifolds \cite{RS1,RS2}, and had been built into
trigonometric ansatz for solutions of the WDVV equations \cite{MH,misha}. The dual prepotentials of the extended-affine Weyl group
orbit spaces should also fall into this class \cite{DZ}.
What is new in this case (for modular Frobenius manifolds) is the invariance property of the function $f$ under $I$ as described in
Example \ref{basicexample}. This is a quite severe restriction on the function $f\,$ and it would be interesting to find further examples of functions with this property. Similarly, it would be of interest to find further examples of elliptic $\vee$-systems.

Some of these results in this paper may also be obtained by other means. For the simple elliptic singularities ${\tilde{E}}_{6,7,8}$ the modular property of the twisted periods may be obtained
directly from the superpotential/unfolding of the underlying singularity. For example, for ${\tilde{E}}_6$ the
superpotential is
\begin{eqnarray*}
f_t({\bf z},{\bf s}) & = & z_1^3 + z_2^3 + z_3^3 + s_8 z_1 z_2 z_3 \\
&&+ s_7 z_1 z_2 + s_6 z_1 z_3 + s_5 z_2 z_3 + s_4 z_1 + s_3 z_2 + s_2 z_3 + s_1\,,
\end{eqnarray*}
where the relationship between the $s_i$ and the flat coordinates $t^i$ have been found explicitly \cite{VW}. The twisted
periods are given by contour integrals \cite{E}
\[
\mathfrak{p}(t,\nu) = \sqrt{h(t^8)} \oint_\gamma f_t^{\nu-1} dx\wedge dy\wedge dz
\]
(note the primitive factor not present in the corresponding formulae for simple singularities). It was noticed in \cite{VW}
that $f_t$ is invariant under the action of $I\,.$ The simple transformation property of $h$ then gives the required
relation
\[
\hat{\mathfrak{p}}=\frac{\mathfrak{p}}{t^N}\,.
\]

Finally, one may calculate how objects such as the isomonodromic $\tau$-functions (denoted $\tau_I$ to avoid confusion with the
variable $\tau$) transform under inversion \cite{me}:
\[
{\hat{\tau}}_I = \frac{\tau_I}{\sqrt{t_1}}\,.
\]
For modular Frobenius manifolds this translates in the following modularity property for $\tau_I\,,$
\[
{\hat{\tau}}_I \left( \frac{{\bf{z}}}{\tau},-\frac{1}{\tau} \right) = \frac{\tau_I({\bf z},\tau)}{\sqrt{\tau}}\,.
\]
For the Chazy example $\tau_I^{-48} = (t^2)^{12} \Delta^3(t^3)\,,$ where $\Delta$ is the classical discriminant
of the underlying elliptic curve, and one may verify the transformation property by using the well-known modular transformation
properties of the discriminant \cite{KS}.

The next step is to calculate higher-order corrections to the dispersive deformations of these flows (a problem also raised, with some conjectures, in
\cite{zhang}), and in particular the higher-order corrections to flows originating from a modular Frobenius manifold. One would
hope that the modularity properties would carry over to the full dispersive hierarchy in some form (recall that a reciprocal
transformation was required for a precise statement of invariance even at zero order). At first order this is entirely
computational, on using the formulae in \cite{DZ2}: the $G$-function is know explicitly \cite{me,KS,iabs3} for a wide range of
modular Frobenius manifolds. An understanding of the modularity properties of the Universal Loop Equation \cite{DZ3} or of
approaches based on singularity theory \cite{GM} for these modular Frobenius manifolds would be of great interest and is potentially
a way to construct new examples of dispersive integrable hierarchies.

\section*{Acknowledgements} IABS would like to acknowledge financial support from the British Council (PMI2 Research Co-operation award)
and EM would like to thank the EPSRC for financial support. We both would like to thank Misha Feigin for various useful conversation
and remarks.

\end{document}